\documentclass{article}
\usepackage{amssymb} 
\usepackage{amsmath}
\usepackage{amsthm}   
\usepackage[pdfpagemode=UseOutlines,colorlinks=true]{hyperref}
\usepackage[pdftex]{graphicx}
\usepackage[pdftex]{color}
\usepackage{natbib}
\usepackage{enumerate}
\usepackage{accents}

\newfont{\smcal}{cmu10 scaled 1200}
\newfont{\handw}{cmmi10 scaled 1200}
\newfont{\handws}{cmmi10 scaled 800}
\newtheorem{Prop}{Proposition}[section]

\newtheorem{Th}[Prop]{Theorem}
\newtheorem{Rm}[Prop]{Remark}
\newtheorem{Def}[Prop]{Definition}

\newtheorem{Cor}[Prop]{Corollary}

\newtheorem{Ex}[Prop]{Example}

\newcommand{\var}{\mbox{\rm Var}}

\DeclareMathOperator{\supp}{supp}

\DeclareMathOperator{\sign}{sign}

\begin{document}
      \title{Intrinsic Means on the Circle: \\Uniqueness, Locus and Asymptotics}
        \author{Thomas Hotz\footnote{supported by DFG CRC 803 and SAMSI 2010 workshop AOOD.}~~and Stephan Huckemann\footnote{supported by DFG HU 1575/2-1 and SAMSI 2010 workshop AOOD.}
\\[6pt]\small Institute for Mathematical Stochastics, University of Goettingen
\\\small Goldschmidtstrasse 7, 37077 Goettingen, Germany
\\\small \{hotz,huckeman\}@math.uni-goettingen.de}
	
        \maketitle

\begin{abstract}
This paper gives a comprehensive treatment of local uniqueness,\linebreak asymp\-to\-tics and numerics for intrinsic means on the circle. It turns out that local uniqueness as well as rates of convergence are governed by the distribution near the antipode. In a nutshell, if the distribution there is locally less than uniform, we have local uniqueness and asymptotic normality with a rate of $n^{-1/2}$. With increased proximity to the uniform distribution the rate can be arbitrarly slow, and in the limit, local uniqueness is lost. Further, we give general distributional conditions, e.g. unimodality, that ensure global uniqueness.
Along the way, we discover that sample means can occur only at the vertices of a regular polygon which allows to compute intrinsic sample means in linear time from sorted data. This algorithm is finally applied in a simulation study demonstrating the dependence of the convergence rates on the behavior of the density at the antipode.
\end{abstract}
\par
\vspace{9pt}

\par
\vspace{9pt}
\noindent {\it AMS 2000 Subject Classification:} \begin{minipage}[t]{6cm}
Primary 62H11\\ Secondary 60F05
 \end{minipage}
\par

\section{Introduction}
 
	The need for statistical analysis of directional data arises in many applications, be it in the study of wind direction, animal migration or geological crack development. An overview of and introduction into this field can be found in \cite{MJ00}. Until today, nonparametric inferential techniques employing the \emph{intrinsic mean}, i.e. the minimizer of expected squared \emph{angular} distances, rest on the assumption, that distributions underlying circular data have no mass in an entire interval opposite to an intrinsic mean. This assumption, however, is not met by any of the standard distributions for directional data, e.g. Fisher-von Mises, Bingham, wrapped normal or wrapped Cauchy distributions.

	The reason for this mathematical assumption lies in the fact that the intrinsic distance is not differentiable for two antipodal points. However, the central limit theorem for intrinsic means on general manifolds derived by \cite{BP05}, cf. also \cite{H_Procrustes_10}, utilizes a Taylor expansion of the \emph{intrinsic variance} by differentiating under the integral sign, i.e. by differentiating the squared intrinsic distance.
	In consequence, this has left the derivation of a central limit theorem along with convergence rates for circular data of most realistic scenarios and distributions an open problem until now.

	Here, we fill this gap and provide for a comprehensive solution. In particular, we show that, 
	if the distribution near the antipode stays below the uniform distribution, asymptotic normality with the $n^{-1/2}$ rate well known from Euclidean statistics remains valid, the asymptotic variance however increases with proximity to the uniform distribution.
	Furthermore, if the distribution at the antipode differs from the uniform distribution only in higher order, then the asymptotic rate is accordingly lowered by this power. If and only if the distribution at the antipode is locally uniform the intrinsic mean is no longer locally unique. Moreover, at an antipode of an intrinsic mean, there may never be more mass than that of the uniform distribution. In particular there can never be a point mass. For these reasons, sample means are always locally unique and we will see that they can only occur on the vertices of a regular polygon. Hence from ordered data, sample means can be computed in linear time. These insights also allow to derive general distributional assumptions such as unimodality under which there is only one local minimizer and hence a unique intrinsic mean. 

	This last result extends the result of \cite{Le98} who guarantees uniqueness for distributions symmetric with respect to a point, being non-increasing functions of the distance to this point, strictly decreasing on a set of positive circular measure, cf. also \cite{KaziskaSrivastava2008}, as well as the very general result of \cite{Afsari10} which in our case of the circle, yields uniqueness if the distribution is restricted to an open half circle.

	In the following Section \ref{Setup:scn} we introduce notation and the concepts of intrinsic means and intrinsic variance. Then we derive in Section \ref{Distr:scn} our first central result for the distribution at the antipode and collect consequences concerning uniqueness and algorithmic methods. In particular, we study the loci of local minimizers and the maximal number of such. Section \ref{CLT:scn} derives various versions of central limit theorems depending on the proximity to the uniform distribution near the antipode. In Section \ref{Simu:scn} we simulate some cases of non-standard asymptotics derived before. We conclude with a discussion of our results and an outlook to large sample asymptotics of data assumed on general manifolds.

\section{Setup}\label{Setup:scn}
	Throughout this paper we consider the unit circle $S^1$ as the interval $[-\pi,\pi)$ with the endpoints identifed. More precisely we equip this interval with the topology generated by the usual topology of $(-\pi,\pi)$ and by all $[-\pi,\epsilon-\pi)\cup (\pi-\epsilon',\pi)$, $0<\epsilon,\epsilon' <\pi$, a base of \emph{neighborhoods of $-\pi$}. On $S^1$ we consider random elements $X_1,\ldots, X_n\stackrel{i.i.d.}{\sim}X$ mapping from a given probability space $(\Omega,\cal A,\mathbb P)$ to $[-\pi,\pi)$, and the \emph{intrinsic distance}
	$ d(\mu,X) := \big|\mu - X -2\pi \nu^{(X)}(\mu)\big|$ from $X$ to a given point $\mu \in S^1$ with
	$$\nu^{(X)}(\mu) := \left\{\begin{array}{rl} 1&\mbox{ if } \mu >0, X \in [-\pi,\mu-\pi)\,,\\-1&\mbox{ if } \mu <0, X \in (\mu+\pi,\pi)\,,\\ 0&\mbox{ else.}\end{array}\right.$$
	In this paper, with the antipodal map $\phi : S^1 \rightarrow S^1$
	$$\phi(x) = \left\{\begin{array}{rl} x-\pi&\mbox{ if }0\leq x<\pi,\\x+\pi&\mbox{ if }-\pi\leq x\leq 0,\end{array}\right.$$
	the \emph{antipodal set}  
	$\widetilde{S} := \{\phi(x): x\in S\}$ of $S\subset S^1$ will play a central role. 

	Letting  $\mathbb E$ denote the usual Euclidean expectation, we  are looking for minimizers $\mu \in [-\pi,\pi)$ of
	$$ V(\mu) := \mathbb E\left(d(\mu,X)^2\right),\mbox{ and } V_n(\mu) := \frac{1}{n}\, \sum_{j=1}^n d(\mu,X_j)^2\,.$$

	\begin{Rm}
	Note that $V$ and $V_n$, being continuous functions on the compact $S^1$, always feature at least one global minimizer.
	\end{Rm}
 
	\begin{Def}
	Every such global minimizer $\mu^*$ of $V$ and $\mu_n$ of $V_n$ is called an \emph{intrinsic population mean} and an \emph{intrinsic sample mean}, respectively. The values $V(\mu^*)$ and $V_n(\mu_n)$ are called \emph{total intrinsic population} and \emph{sample variance}, respectively. 
	\end{Def}

	Obviously, global minimizers need not be unique, and there may be local but non-global minimizers, as the examples below will show.
	
	Denoting the usual average by $\bar{X} = \frac{1}{n}\, \sum_{j=1}^n X_j$ we have that $-\pi \leq \bar{X} < \pi$. Since $V(\mu) \leq \mathbb E \vert \mu - X \vert^2$ with equality for $\mu = 0$, the latter functional being minimized by $\mu = \mathbb E X$, we see that $\mu^* = 0$ locally minimizing $V$ implies $\mathbb E X = 0$; cf. also \citet[Section VIII.9]{KN69} as well as \cite{Ka77}. The classical \emph {Euclidean Central Limit Theorem} hence gives
	\begin{eqnarray}\label{Eucl-CLT:eq}
	 \sqrt{n}\, \bar{X} &\stackrel{D}{\to}& {\cal N}(0,\sigma^2)
	\end{eqnarray}	
	with the Euclidean variance $\var(X) = \sigma^2$. If $\mu^* = 0$ is also a global minimizer of $V$, i.e. an intrinsic mean, then the Euclidean variance agrees with the total intrinsic variance 
	$$V(0) = \sigma^2\,.$$
	If we were interested whether any other $\mu\in S^1$ were an intrinsic mean, we could apply a simple rotation of the circle to reduce this to the case $\mu = 0$ again. Without loss of generality we can hence restrict our attention to this special $\mu$; recall that its antipode is $\phi(0) = -\pi$.

\section{Local and Global Minimizers}
\label{Distr:scn} 

\subsection{The Distribution Near The Antipode}

Here is our first fundamental Theorem.

\begin{Th}
\label{charact:thm}
If $\mu^* = 0$ locally minimizes $V$, then:
\begin{enumerate}[(i)]
\item $\mathbb P_X(-\pi) = 0$, i.e. there is no point mass opposite an intrinsic mean.
\item If additionally there is some $\delta >0$ s.t. $\mathbb P_X$ restricted to $(-\pi, -\pi + \delta)$ features a continuous density $f$ w.r.t. Lebesgue measure, then $f(-\pi+) = \lim_{\mu \downarrow -\pi} f(\mu) \leq \frac{1}{2\pi}$; similarly $f(\pi-) =\lim_{\mu \uparrow \pi} f(\mu) \leq \frac{1}{2\pi}$ for a continuous density $f$ on $(\pi-\delta,\pi)$. If both $f(-\pi+) < \frac{1}{2\pi}$ and $f(\pi-) < \frac{1}{2\pi}$, then $\mu^*$ is locally unique.
\item In case $f(-\pi) = \frac{1}{2\pi}$, $f$ being continuous in a neighborhood of $-\pi$, assume that there is some $\delta>0$ s.t. $f$ is $k$-times continuously differentiable on $[-\pi,-\pi+\delta)$ (i.e. with existing left limits) and $\tilde k$-times continuously differentiable on $(\pi -\delta,\pi]$ (i.e. with existing right limits), where $k$ and $\tilde k$ are chosen minimal s.t. $f^{(k)}(-\pi+) = \lim_{\mu \downarrow -\pi} f^{(k)}(\mu) \neq 0$ as well as $f^{(\tilde k)}(\pi-) = \lim_{\mu \uparrow \pi} f^{(\tilde k)}(\mu) \neq 0$. Then, $f^{(k)}(-\pi+) <0$ as well as $(-1)^{\tilde k}f^{(\tilde k)}(\pi-) <0$, $\mu^*$ is locally unique, and 
for $\delta > 0$ small enough
$$ \mathbb P\{-\pi \leq X \leq - \pi + \delta\} = \frac{\delta}{2 \pi} + \frac{\delta^{k+1}}{(k+1)!} f^{(k)}(-\pi+) + o(\delta^{k+1}), $$
as well as
$$ \mathbb P\{\pi - \delta \leq X < \pi \} = \frac{\delta}{2 \pi} + \frac{\delta^{\tilde k+1}}{(\tilde k+1)!} (-1)^{\tilde k} f^{(\tilde k)}(\pi-) + o(\delta^{\tilde k+1}). $$
\end{enumerate}
\end{Th}
\begin{proof}
For any $\mu > 0$, we have
\begin{align*}
V(\mu) - V(0) &= \int_{-\pi+\mu}^\pi (\mu - x)^2\,d\mathbb P_X(x) + \int_{-\pi}^{-\pi+\mu} (\mu - x-2\pi)^2\,d\mathbb P_X(x) - \int_{-\pi}^\pi x^2\,d\mathbb P_X(x)
\\&= \int_{-\pi}^\pi \underbrace{\big( (\mu-x)^2 - x^2 \big)}_{=\mu^2 - 2\mu x}\,d\mathbb P_X(x) + \int_{-\pi}^{-\pi+\mu} \big( -4\pi(\mu - x) +4\pi^2 \big)\,d\mathbb P_X(x)
\\&= \mu^2 - 2 \mu \underbrace{\mathbb EX}_{=0} - 4\pi \int_{-\pi}^{-\pi+\mu} \underbrace{( -\pi+ \mu - x)}_{\geq 0}\,d\mathbb P_X(x)
\\&\leq \mu^2 - 4\pi\,\mu\, \mathbb P_X(-\pi)
\end{align*}
which for $\mu$ small enough becomes negative if $\mathbb P_X(-\pi) > 0$, whence (i) follows.

Now denote the (shifted) cumulative distribution function (c.d.f.) of $Y=X+\pi$ by 
\begin{eqnarray}\label{F:def}
F(\mu) &=&\mathbb P\{0 \leq Y \leq \mu\}~=~\mathbb P\{-\pi \leq X \leq -\pi + \mu\}
\end{eqnarray}
for $\mu \geq 0$ to obtain 
\begin{equation*}
\int_{-\pi}^{-\pi+\mu}( -\pi+ \mu - x)\,d\mathbb P_X(x) = \int_0^\mu (\mu - y) F'(y) dy
= 0 F(\mu) - \mu \underbrace{F(0)}_{=0} + \int_0^\mu F(y) dy,
\end{equation*}
$F'$ being understood in a distributional sense.
Noting that $\mu^2 = 4\pi \int_0^\mu \frac{x}{2\pi} dx$ where $\frac{x}{2\pi}$ is the (shifted) c.d.f. of the uniform distribution on $[-\pi,\pi)$, we see that
\begin{equation}
\label{charact:eq}
V(\mu) \geq V(0) \text{ iff } \int_0^\mu \Big(\frac{x}{2\pi} - F(x)\Big) dx \geq 0,
\end{equation}
where equality holds simultaneously, too.

Thus, under the assumptions of (ii), we may use a 2nd order Taylor expansion to obtain for $\mu\geq 0$ small enough
\begin{align*}
0 \leq \int_0^\mu \Big(\frac{x}{2\pi} - F(x)\Big) dx &= \frac{\mu^2}{4\pi} -\frac{\mu^2}{2} F'(0) + o(\mu^2),
\end{align*}
hence $f(-\pi+) = F'(0+) \leq \frac{1}{2\pi}$. If this inequality is strict, $V(\mu) > V(0)$ follows for $\mu$ small enough. With the analogous argument for $\mu < 0$ this yields the assertion in (ii).

Finally, (iii) is obtained by a Taylor expansion as well, namely (for $\mu>0$)
$$ \frac{\mu}{2\pi} - F(\mu) = - \frac{\mu^{k+1}}{(k+1)!} f^{(k)}(-\pi+)  + o(\mu^{k+1}), $$
which after integration, noting that uniqueness implies a sharp inequality, gives
$$ 0 \leq \int_0^\mu \Big(\frac{x}{2\pi} - F(x)\Big) dx = - \frac{\mu^{k+2}}{(k+2)!} f^{(k)}(-\pi+) + o(\mu^{k+2}), $$
whence $f^{(k)}(-\pi+) < 0$ follows as well as local uniqueness. The case $\mu < 0$ is again treated analogously.
\end{proof}

\subsection{Consequences for Uniqueness, Loci of Local Minimizers and Algorithms}

From Theorem \ref{charact:thm} we obtain at once a necessary and sufficient condition such that a minizer of $V$ is locally unique.

\begin{Cor}\label{V-const:Cor}
By \eqref{charact:eq}, $V$ is constant on some interval $S\subset S^1$ iff the probability distribution restricted to the antipodal interval $\widetilde{S}$ 
has constant density $\frac{1}{2\pi}$ there. In particular, suppose that $\mu^* \in S^1$ is a local minimizer of $V$. Then $\mu^*$ is locally unique iff there is no interval $\mu^*\in S\subset S^1$ such that the distribution restricted to $\widetilde{S}$ is uniform. 
\end{Cor}

The following is a generalization of a result by Rabi Bhattacharya (personal communication from 2008):

\begin{Cor}\label{piece:Cor}
If the distribution of $X$ has a density w.r.t. Lebesgue measure which is composed of finitely many pieces, each being analytic up to the interval boundaries, then any local minimizer of $V$ is locally unique unless the density is constant $\frac{1}{2\pi}$ on some interval.
\end{Cor}
\begin{proof}
This follows immediately since an analytic function is constant unless one of its derivatives is non-zero.
\end{proof}

\begin{Prop}\label{unique-minimizer:prop}
Consider the distribution of $X$, decomposed into the part $\lambda$ which is absolutely continuous w.r.t. Lebesgue measure, with density $f$, and the part $\eta$ singular to Lebesgue measure. Let $S_1,\ldots,S_k$ be the distinct open arcs on which $f<\frac{1}{2\pi}$, assume they are all disjoint from $\supp \eta$, and that $\{x \in S^1\, :\, f(x) = \frac{1}{2\pi}\}$ is a Lebesgue null-set. Then $X$ has at most $k$ intrinsic means and every $\widetilde{S_j} $ contains at most one candidate.
\end{Prop}

\begin{proof}
Suppose that $\mu^* = 0$ is a local minimizer of $V$. By hypothesis and virtue of Corollary~\ref{V-const:Cor} there is an open arc $S_1\ni-\pi$ followed (going into positive $x$-direction) by a closed  arc $T_1$ such that $f(x) <1/(2\pi)$ for $x \in S_1\setminus\{-\pi\}$ and $f(x) > 1/(2\pi)$ for $x\in T_1$ a.e.. Let $\widetilde{S_1}= (-\delta',\delta)$ and $\widetilde{T_1}=[\delta,\delta'']$ for some $0<\delta',\delta$ and $\delta''>\delta$. It suffices to show that no $0<\mu <\delta''$ can be another minimizer of $V$.

To this end with $F$ from (\ref{F:def}) consider
\begin{equation*}
G(\mu) = \frac{\mu}{2\pi} - F(\mu) \text{\quad and\quad } H(\mu) = \int_0^\mu G(x) dx = \int_0^\mu \left(\frac{x}{2\pi} - F(x)\right) dx\, ,
\end{equation*}
cf. \eqref{charact:eq} and recall that $\mu$ is a minimizer of $V$ iff $H(\mu)\geq 0$. By construction, $G'$ is positive on $(0,\delta)$ and it is negative on $[\delta,\delta']$ a.e. Hence $G$ is continuous and $G(0)=0<G(\mu)$ for small $\mu >0$. Hence, $H(\mu)>0$ for all $0<\mu\leq \delta$. Arguing again with Corollary~\ref{V-const:Cor} that there cannot be a minizer of $V$ if its antipode carries more density than the uniform density, gives that $\delta <\mu <\delta''$ cannot be a minimizer either, completing the proof.\end{proof}

We note two straightforward consequences.

\begin{Cor}~\label{candidates:Cor}
\begin{enumerate}[(i)]
 \item Every population mean of a unimodal distribution is globally unique.
\item
If the distribution of $X$ is composed of $k < \infty$ point masses at distinct locations, then $V$ has at most $k$ local minimizers, each being locally unique; for each interval formed by two neighboring point masses there is at most one local minimizer in the interior of the antipodal interval.
\item
In particular, any intrinsic sample mean is locally unique.
\end{enumerate}
\end{Cor}

Curiously, candiates for intrinsic sample means are very easy to obtain from one another.
\begin{Cor}\label{sample:cor} For a sample $X_1,\ldots,X_n$ the candidates for minimizers of $V_n$ described in Proposition~\ref{unique-minimizer:prop} and Corollary~\ref{candidates:Cor} form the vertices of a regular $n$-polgyon. If $(X_1,\ldots,X_n)$ is continuously distributed on $(S^1)^n$, then there is a.s. one and only one intrinsic sample mean.\end{Cor}

\begin{proof} W.l.o.g. assume that the sample is ordered, i.e. $-\pi \leq X_1 \leq\ldots \leq X_n < \pi$. 
We consider for $\mu\geq 0$ the case that $X_i < \mu - \pi < X_{i+1}$ for $1 \leq i \leq n -1$ or $\mu - \pi < X_1$ for $i = 0$. Note that equalities $X_i = \mu - \pi$ etc. are excluded by Theorem~\ref{charact:thm}(i); also, observe that $\mu-\pi >X_n$ for $i=n$ cannot lead to a local minimum at $\mu$. Then, setting the first sum to zero in case $i=0$,
\begin{align}
 V_n(\mu) &= \frac{1}{n}\,\left(\sum_{j=1}^i(X_j - \mu +2\pi)^2 + \sum_{j=i+1}^n(X_j - \mu)^2\right)\notag\\
	&= \frac{1}{n}\,\sum_{j=1}^n(X_j - \mu)^2 - \frac{4\pi}{n} \sum_{j=1}^i(\mu-X_j - \pi). \label{Vpos:eq}
\end{align}
This is minimal for $\mu^{(i)} = \bar{X} + \frac{2\pi i}{n}$. However, only if $X_i +\pi < \mu^{(i)} < X_{i+1}+\pi$, this minimum corresponds to a local minimum of $V_n$.
Similarly, for $\mu < 0$ and $X_i < \mu + \pi < X_{i+1}$ for $1 \leq i \leq n -1$ or $X_n < \mu + \pi$ for $i = n$, we get
\begin{align}
 V_n(\mu) &= \frac{1}{n}\,\left(\sum_{j=1}^i(X_j - \mu )^2 + \sum_{j=i+1}^n(X_j - \mu-2\pi)^2\right)\notag\\
	&= \frac{1}{n}\,\sum_{j=1}^n(X_j - \mu)^2 - \frac{4\pi}{n} \sum_{j=i+1}^n(X_j - \mu - \pi), \label{Vneg:eq}
\end{align}
which is minimal for $\mu^{(i)} = \bar{X} - \frac{2\pi (n-i)}{n}$. Again, note that $X_1>\mu+\pi$ for $n=0$ cannot lead to a local minimum at $\mu$. 

With $\bar X_i = \frac{1}{i} \sum_{j=1}^iX_j$ and $\underaccent{\bar}{X}_i = \frac{1}{n-i} \sum_{j=i+1}^nX_j$, we obtain for a local minimizer $\mu^{(i)}$ that $V_n(\mu^{(i)}) = v_{n,i}$ where
\begin{equation}
v_{n,i} = \frac{1}{n}\sum_{j=1}^n(X_j - \bar X)^2 + \begin{cases}
- \big( \frac{2\pi i}{n} \big)^2 + \frac{4 \pi i}{n} \big(\pi + \bar X_i - \bar X \big), & \mu^{(i)} \geq 0, \\
- \big( \frac{2\pi (n-i)}{n} \big)^2 + \frac{4 \pi(n-i)}{n} \big(\pi - \underaccent{\bar}{X}_i + \bar X \big), & \mu^{(i)} < 0, \\
\end{cases} \label{polygon:eq}
\end{equation}
whence $V_n(\mu^{(i)}) = V_n(\mu^{(j)})$ for $i\neq j$ implies that there are $a,b,c,d\in\mathbb Z$, with $a \neq 0$ or $b \neq 0$ s.t. $a \bar X_i + b \bar X_j + c \bar X + d\pi = 0$, the probability of which is zero for continuously distributed data. 
\end{proof}

\begin{Rm}
That intrinsic sample means of continuous distributions are a.s. globally unique has also been observed by \citet[Remark~2.6]{BP03}.
\end{Rm}

\begin{Rm}\label{num:rm}
Corollary~\ref{sample:cor} has application for algorithmically determining an intrinsic sample mean: in order to do so, determine the minimizers of $V_n(\mu^{(i)})$ for $\mu^{(i)} \equiv \bar{X} + \frac{2\pi i}{n} \mod 2 \pi$, $i=1, \dots, n$; this requires as many steps as there are data points in the sample. In fact, \eqref{polygon:eq} easily leads to an implementaion requiring $O(n)$ time for computing the intrinsic mean(s) of a sorted sample. An example is shown in Figure~\ref{figv}; note that not all points on the polygon form candidates $\mu^{(i)}$: if $\bar X > \frac{2\pi i}{n}$ then $\bar X - \frac{2\pi i}{n}$ cannot be a minimizer. Also note that $V_n(\mu^{(i)})$ and $v_{n,i}$ only agree if $\mu^{(i)}$ is a local minimizer.

For the computation of population means, Proposition~\ref{unique-minimizer:prop} allows for a similar procedure: compute in each interval with density less than $1/(2\pi)$ the unique minizer. 
\end{Rm}

\begin{figure}[!tbp]
\centering
\includegraphics[width=0.8\textwidth]{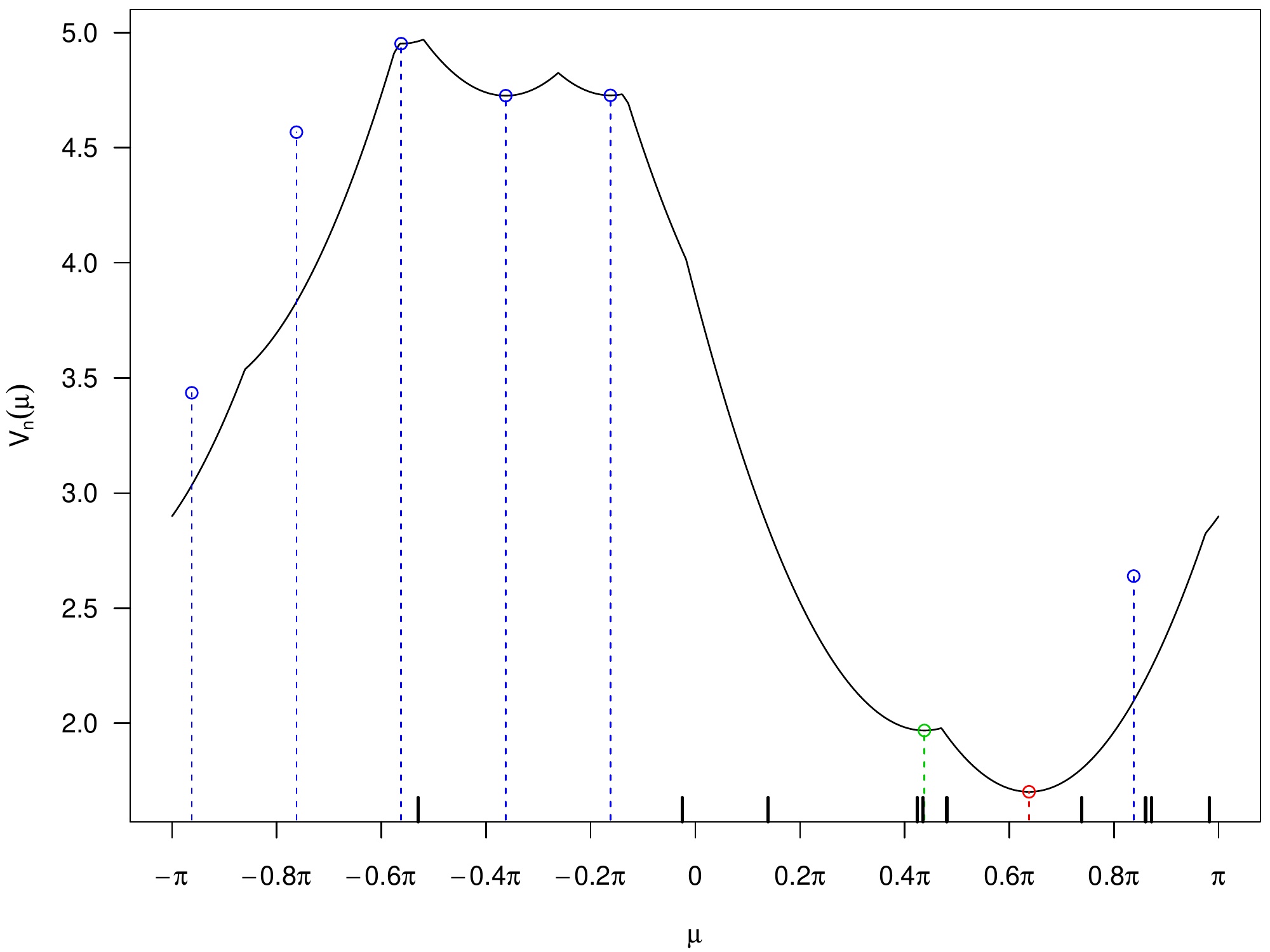}
\caption{\label{figv}Numerical example for determining a sample's intrinsic mean; see Remark\-~\ref{num:rm}: $n=10$ observations $X_1, \dots, X_{10}$ from a wrapped normal distribution are indicated as vertical strokes at the bottom; the curve depicts their $V_n(\mu)$; vertical dashed lines give the $\mu^{(i)}$ lying on a regular polygon from the proof of Corollary~\ref{sample:cor}, the corresponding values $v_{n,i}$ from \eqref{polygon:eq} are indicated by circles; the average $\bar X$ of the data is shown in green, the intrinsic mean in red.}
\end{figure}

Finally, we give an illustration to Corollary~\ref{V-const:Cor}. 

\begin{Ex}\label{const:ex}
Suppose that $X$ is uniformly distributed on $[-\pi,-\pi + \delta\pi]\cup [\pi -\delta\pi,\pi)$ with $0 \leq \delta \leq \frac{1}{2}$ and total weight $0\leq \alpha\delta\leq 1$, i.e. $0 \leq \alpha \leq \delta^{-1}$ giving a density of $\frac{\alpha}{2\pi}$ near $\pm \pi$, and with a point mass of weight $1 - \alpha\delta$ at $0$. Then in case of $\alpha=1$, by  Corollary~\ref{V-const:Cor}, $V(\mu)$ is constant for $-\delta\pi\leq \mu \leq \delta\pi$, and $[-\delta\pi,\delta\pi]$ is precisely the set of intrinsic means. Moreover for $\alpha>1$, $\{-\alpha\delta\pi,\alpha\delta\pi\}$ is the set of intrinsic means whereas for $\alpha < 1$, $0$ is the unique intrinsic mean. 
\end{Ex}
\begin{proof}
Indeed, for $0\leq \mu\leq \delta\pi$ we have
\begin{align*}
 V(\mu) &= (1-\alpha\delta)\,\mu^2 + \,\frac{\alpha}{2\pi} \left(\int_{\pi-\delta\pi}^{\pi+\mu}(x-\mu)^2\,dx + \int_{-\pi+\mu}^{-\pi + \delta\pi}(x-\mu)^2\,dx\right)\\
&= (1-\alpha\delta)\,\mu^2 + \,\frac{\alpha}{6\pi} \Big(\underbrace{ \pi^3 - \big((1-\delta)\pi-\mu\big)^3 + \big(-(1-\delta)\pi-\mu\big)^3 + \pi^3}_{=2\pi^3-2\,(1-\delta)^3\pi^3\,- 2\cdot 3 \,(1-\delta)\pi\,\mu^2}\Big)\\
&= \frac{\alpha}{3}(3\delta - 3\delta^2 + \delta^3) \pi^2 + \left(1-\alpha\right)\mu^2\,
\end{align*}
which is constant in $\mu$ for $\alpha=1$, minimal for $\mu =0$ in case of $\alpha<1$ and minimal for $\mu =\delta\pi$ in case of $\alpha>1$.
On the other hand for $\delta\pi\leq \mu \leq \pi$ we have
\begin{align*}
 V(\mu) &= (1-\alpha\delta)\,\mu^2 + \,\frac{\alpha}{2\pi} \int_{\pi-\delta\pi}^{\pi+\delta\pi}(x-\mu)^2\,dx\\
&= (1-\alpha\delta)\,\mu^2 + \,\frac{\alpha}{6\pi} \Big(\underbrace{\big((1+\delta)\pi-\mu\big)^3 - \big((1-\delta)\pi-\mu\big)^3}_{=((1+\delta)^3-(1-\delta)^3)\pi^3 - 3((1+\delta)^2-(1-\delta)^2)\pi^2\,\mu + 2\cdot 3\,\delta\pi\,\mu^2}\Big)\\
&= \frac{\alpha}{3} (3 \delta + \delta^3) \pi^2 - 2 \alpha\delta \pi\, \mu + \mu^2\\
&= V(0) +\delta^2\alpha(1-\alpha)\pi^2+ (\mu - \alpha\delta\pi)^2
\end{align*}
which is minimal for $\mu = \alpha\delta\pi$. In case of $\alpha =1$ this minimum agrees with $V(0)$, in case of $\alpha < 1$ it is larger than $V(0)$, and in case of $\alpha > 1$ it is smaller than $V(\delta\pi)$. 
\end{proof}

\section{Asymptotics}\label{CLT:scn}

The strong law of large numbers established by \cite{Z77} for minimizers of squared quasi-metrical distances applied to the circle with its intrinsic metric, which renders it a compact space for which the sequence of $\mu_n$ necessarily features an accumulation point, gives the following theorem; cf. also \citet[Theorem~2.3(b)]{BP03}
\begin{Th}\label{SLLN:thm}
If $\mu^*$ is the unique minimizer of $V$ and $(\mu_n)_{n\in\mathbb N}$ a measurable choice of minimizers of $V_n$, then $\mu_n \rightarrow \mu^*$ a.s.

More generally, if $E_n$ denotes the set of intrinsic sample means, and $E$ the set of intrinsic population means, then
\begin{equation}\label{Ziez:strong-law}
\mathop{\bigcap}_{n=1}^\infty\overline{\bigcup_{k=n}^\infty E_n}\ \subset\ E \text{ a.s.}
\end{equation}
\end{Th}

We now characterize the asymptotic distribution of $\mu_n$ under similar assumptions as in Theorem~\ref{charact:thm}, though additionally requiring that the locally unique intrinsic mean is in fact globally unique.

	\begin{Th}\label{CLT:thm} Assume that the distribution of $X$ restricted to some neighborhood of $-\pi$ features a continuous density $f$, has Euclidean variance $\sigma^2$ and that $\mu^* = 0$ is its unique intrinsic mean. Then the following assertions hold for the intrinsic sample mean $\mu_n$ of $X_1,\ldots, X_n\stackrel{i.i.d.}{\sim}X$:
	\begin{enumerate}
	 \item[(i)] If $f(-\pi) < \frac{1}{2\pi}$ then
	$$\sqrt{n} \,\mu_n \stackrel{D}{\to} {\cal N}\left(0,\frac{\sigma^2}{\big(1-2\pi f(-\pi)\big)^2}\right)\,.$$
	\item[(ii)] If $f(-\pi) = \frac{1}{2\pi}$, and if $f$ is  $(k-1)$-times continuously differentiable in a neighborhood $U$ of $-\pi$ with these $k-1$ derivatives vanishing at $-\pi$ while $f$ is even $k$-times continously differentiable in $U\setminus\{-\pi\}$ with
	$$ 0\neq f^{(k)}(\pi-)=\lim_{\mu\uparrow \pi} f^{(k)}(\mu) = - \lim_{\mu\downarrow -\pi} f^{(k)}(\mu)=-f^{(k)}(-\pi+) < \infty$$
	then
	$$\sqrt{n} \,\sign(\mu_n)\, \vert\mu_n\vert^{k+1} \stackrel{D}{\to} {\cal N}\left(0,\frac{\sigma^2 \big((k+1)!\big)^2}{\big(2\pi f^{(k)}(-\pi+)\big)^2}\right)\,.$$
	\end{enumerate}
	\end{Th}

	\begin{proof} With the indicator function 
	$$\chi_A(X) = \left\{\begin{array}{rl}1&\mbox{ if } X\in A\\0&\mbox{ if } X\not \in A\end{array}\right.$$
	we have that 
	$$\nu^{(X)}(\mu) = \left\{\begin{array}{rl}\chi_{[-\pi,\mu-\pi)}(X_j)&\mbox{ if }\mu >0,\\-\chi_{(\mu+\pi,\pi)}(X_j)&\mbox{ if }\mu <0\end{array}\right.$$
	In consequence of Theorem~\ref{charact:thm}, if $\mu_n$ is an intrinsic sample mean, a.s. none of the $X_j$ ($j=1,\ldots,n$) can be opposite of $\mu_n$. Since the sample mean $\mu_n$ minimizes $V_n(\mu)$, we have hence with the well defined derivative $\frac{d~}{d\mu} V_n(\mu_n)$ that
	\begin{equation}\label{sample_mean:eq}0 = \frac{1}{2}\, \frac{d~}{d\mu}V_n(\mu_n) = \left\{\!\begin{array}{ll} \mu_n - \bar{X} - 2\pi \frac{1}{n}\,\sum_{j=1}^n \chi_{[-\pi,\mu_n-\pi)}(X_j)&\mbox{ for }\mu_n \geq 0\,, \\
	\mu_n - \bar{X} + 2\pi \frac{1}{n}\,\sum_{j=1}^n \chi_{(\mu_n+\pi,\pi)}(X_j)&\mbox{ for }\mu_n <0\,,
	                                         \end{array}\right.
	\end{equation}
	cf. \eqref{Vpos:eq} and \eqref{Vneg:eq}.
	Under the assumptions of (i) above let us now compute
	\begin{eqnarray*}
	 \mathbb E \left(\chi_{[-\pi,\mu -\pi)}(X)\right) &=& \int_{-\pi}^{\mu-\pi}f(x)\,dx~=~
	\mu f(-\pi) + o(\mu) 
	\end{eqnarray*}
	in case of $\mu >0$ and similarly
	\begin{eqnarray*}
	 \mathbb E \left(\chi_{(\mu+\pi,\pi)}(X)\right) 
	&=& -\mu f(-\pi) + o(\mu)\, 
	\end{eqnarray*}
	in case of $\mu<0$. In consequence, using that the variance of these Bernoulli variables is less or equal than their expectation, we get
	\begin{eqnarray*} \frac{1}{n}\sum_{j=1}^n \chi_{[-\pi,\mu-\pi)}(X_i)
	 &=& \mu f(-\pi) + O_{P}\left(\frac{\sqrt{\mu}}{\sqrt{n}}\right) + o(\mu),\\
	\frac{1}{n}\sum_{j=1}^n \chi_{(\mu+\pi,\pi)}(X_i)
	 &=& -\mu f(-\pi) + O_{P}\left(\frac{\sqrt{-\mu}}{\sqrt{n}}\right) + o(\mu)
	\end{eqnarray*}
	for $\mu>0$ and $\mu<0$, respectively, the bounds for $O_P(\sqrt{\mu/n})$ being uniform in $\mu$. In conjunction with (\ref{sample_mean:eq}), and using the SLLN for $\mu_n$ (Theorem~\ref{SLLN:thm}), i.e. $\mu_n = o_P(1)$, we obtain
	\begin{equation*}
	\sqrt{n} \Big( \big(1 - 2\pi f(-\pi) \big) \mu_n - \bar X \Big) = o_P(1)\,.
	\end{equation*}
This gives assertion (i).

	Under the assumptions of (ii) we get, noting that by Theorem~\ref{charact:thm}(iii)\linebreak $\sign(f^{(k)}(-\pi+)) = -1$ while $\sign(f^{(k)}(\pi-)) = (-1)^{k+1}$,
	\begin{eqnarray*}
	 \left.\begin{array}{ll}\mathbb E \left(\chi_{[-\pi,\mu -\pi)}(X)\right)\\\mathbb E \left(\chi_{(\mu +\pi,\pi)}(X)\right)\end{array}\right\}
	&=& \frac{|\mu|}{2\pi} + \frac{\vert \mu \vert^{k+1}}{(k+1)!} f^{(k)}(-\pi+) + o(\mu^{k+1})\, .
	\end{eqnarray*}
	In consequence, as above, 
	we infer from (\ref{sample_mean:eq}) that
	$$\sqrt{n}\Bigg( 2\pi\,\frac{\sign(\mu_n) \vert\mu_n\vert^{k+1}}{(k+1)!}\, f^{(k)}(-\pi+) + \bar X \Bigg) = o_P(1)\, ,$$
	which gives assertion (ii).
	\end{proof}

	\begin{Rm}\label{rate-mu-n:cor} We note that under the assumptions in Theorem~\ref{CLT:thm}, namely that $f$ differs from the uniform distribution at $-\pi$ for the first time in its $k$-th derivative there, then
	the convergence rate of $\mu_n$ is precisely $n^{-\frac{1}{2(k+1)}}$.

Comparing with \eqref{Eucl-CLT:eq}, we see that the asymptotic distribution of $\bar X$ is more concentrated than the one of the intrinsic mean unless $f(-\pi) = 0$, the intrinsic mean exhibiting slower convergence rates than $\bar X$ if $f(-\pi) = \frac{1}{2\pi}$.
	\end{Rm}

\section{Simulation}\label{Simu:scn}

For illustration of the theoretical results we consider here examples exhibiting different convergence rates: we generalize the density from Example~\ref{const:ex} to behave like a polynomial of order $k$ near $\pm \pi$. To be precise, we assume that the distribution of $X$ is composed of a point mass at 0 with weight $1 - \alpha\delta$ with $0 \leq \alpha \leq \delta^{-1}$, $0 \leq \delta \leq \frac{1}{2}$, and of a part absolute continuous w.r.t. Lebesgue measure with density $g$ where $g(-\pi + x) = g(\pi - x) = f(x)$ for $0 \leq x \leq \pi$, and
\begin{equation*}
f(x) = \begin{cases}
\frac{\alpha}{2\pi} &\text{for } 0 \leq \vert x \vert \leq \delta \pi, \\
0 & \text{else}
\end{cases}
\end{equation*}
for $k = 0$, while for $k > 0$
\begin{equation*}
f(x) = \begin{cases}
\frac{\alpha}{2\pi} - \frac{\alpha}{4\pi} \Big(\frac{\vert x \vert}{\delta\pi}\Big)^k &\text{for } 0 \leq \vert x \vert \leq \delta \pi, \\
\frac{\alpha}{4\pi} \Big(2 - \frac{\vert x \vert}{\delta\pi}\Big)^k &\text{for } \delta\pi \leq \vert x \vert \leq 2\delta \pi, \\
0 & \text{else.}
\end{cases}
\end{equation*}
Note that $f(-\pi) = \frac{\alpha}{2\pi}$ while $\int_{-\pi}^\pi f(x) dx = \alpha\delta$. We simulated several examples with parameters given in Table~\ref{tabsim}, the corresponding densities are shown in Figure~\ref{figdens}.

\begin{table}[!tbp]
\centering
\begin{tabular}{l|cccl}
 & $\alpha$ & $k$ & $\delta$ & color \\ 
\hline
case 0a & 0.9 & 0 & 0.4 & blue \\ 
case 0b & 1 & 0 & 0.4 & red \\ 
case 1a & 0.9 & 1 & 0.4 & green \\ 
case 1b & 1 & 1 & 0.4 & brown \\ 
case 2 & 1 & 2 & 0.4 & violet \\ 
case 3 & 1 & 3 & 0.4 & purple \\ 
\end{tabular}
\caption{\label{tabsim}Parameters for the simulations, and respective colors used in Figures~\ref{figdens} and~\ref{figmad}.}
\end{table} 

Example~\ref{const:ex} is the special case for which $k = 0$; in particular, for case~0b, we computed $10,000$ intrinsic means, each of which was based on $n = 10,000$ i.i.d. observations, a histogram of these intrinsic means is shown in Figure~\ref{fighist}. There, the distribution of the intrinsic sample mean for case~0b appears to be composed of two parts: an essentially constant density over $[-\delta\pi,\delta\pi]$, the set of the intrinsic population mean, and two peaks with their modes located close to the interval's endpoints. Their presence can be explained as follows: approximately with probability one half, we observe less than $(1 - \delta) n$ zeros, whence there is too little mass at $0$ to keep the intrinsic sample mean in the interval $[-\delta\pi,\delta\pi]$ but for $n$ large enough there will with large probability still be many zeros so that the intrinsic sample mean cannot move far away from that interval. According to Theorem~\ref{SLLN:thm}, these peaks' locations converge to the interval's boundary when $n \rightarrow \infty$. In fact, our simulations suggest that we have $\cap_{n=1}^\infty\overline{\cup_{k=n}^\infty E_n} \cap A \neq \emptyset$ with positive probabiliy for any $A \subset [-\delta\pi,\delta\pi]$ having non-zero Lebesgue measure.

\begin{figure}[!ptb]
\centering
\includegraphics[width=0.8\textwidth]{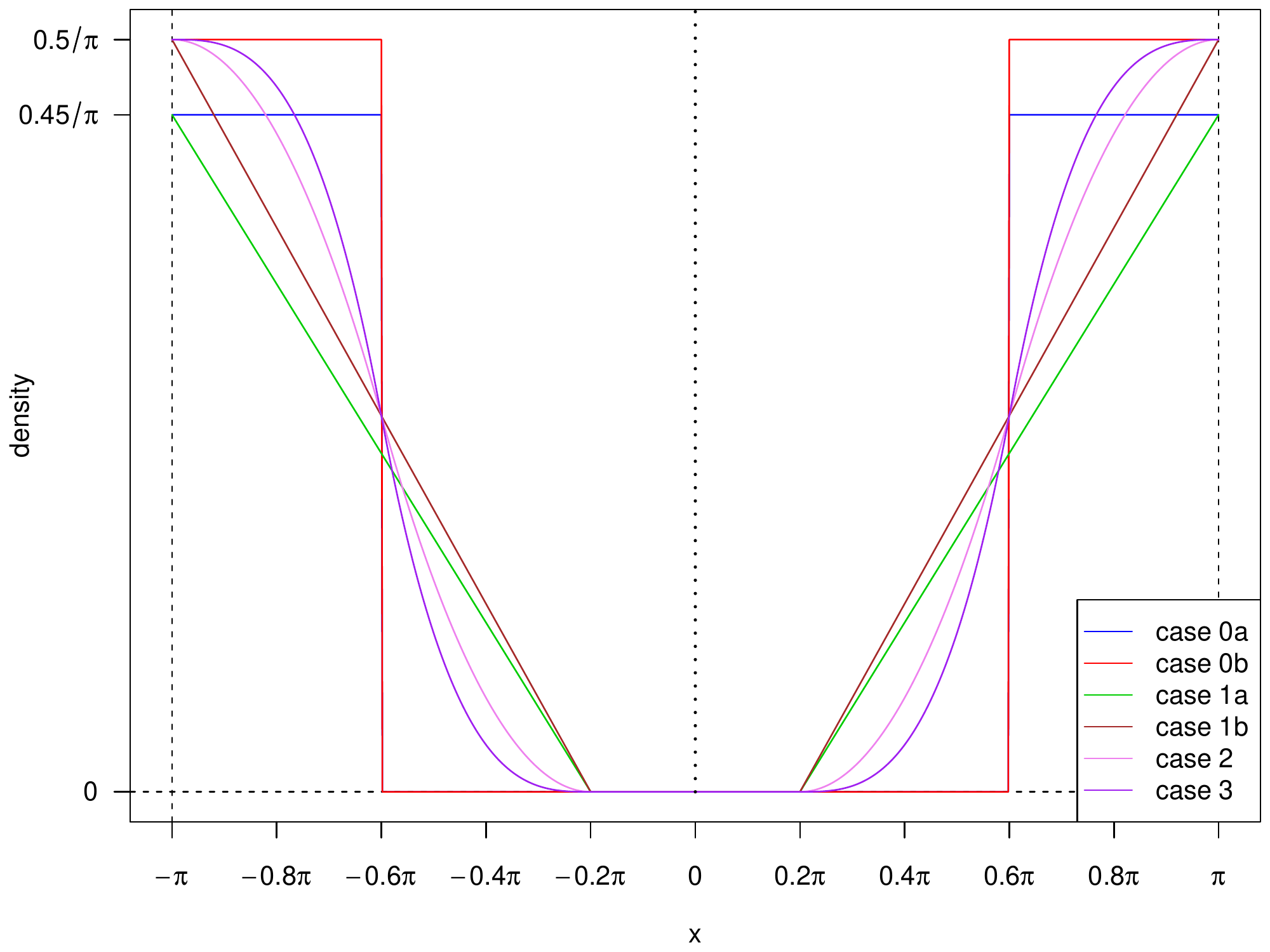}
\caption{\label{figdens}Densities of the simulated distributions with parameters in Table~\ref{tabsim}; the dotted vertical line in the centre indicates the point mass at $0$.}
\end{figure}

\begin{figure}[!ptb]
\centering
\includegraphics[width=0.8\textwidth]{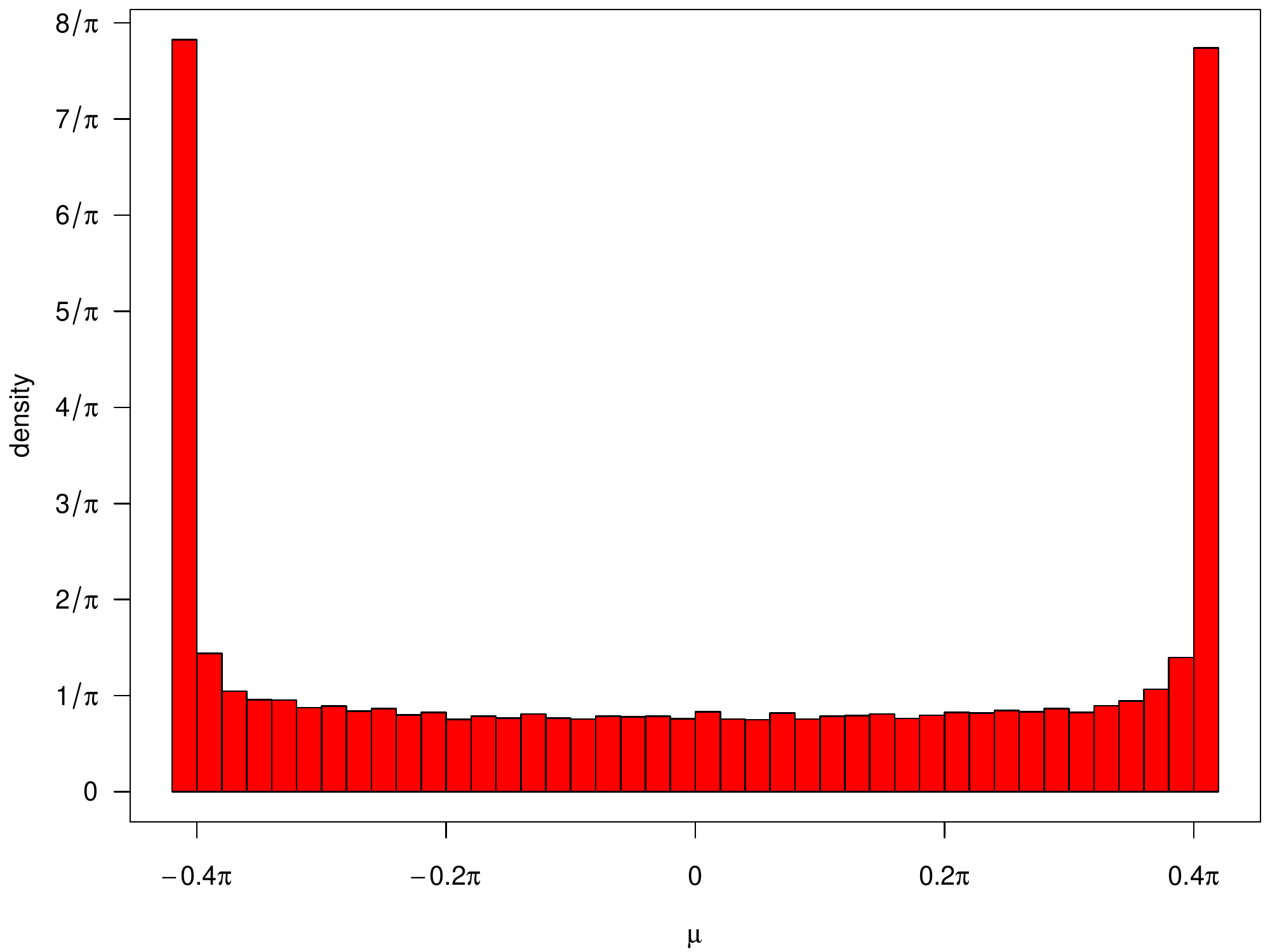}
\caption{\label{fighist}Histogram of $10,000$ intrinsic means each based on $n = 10,000$ i.i.d. draws from case 0b in Table~\ref{tabsim}.}
\end{figure}

\begin{figure}[!ptb]
\centering
\includegraphics[width=0.8\textwidth]{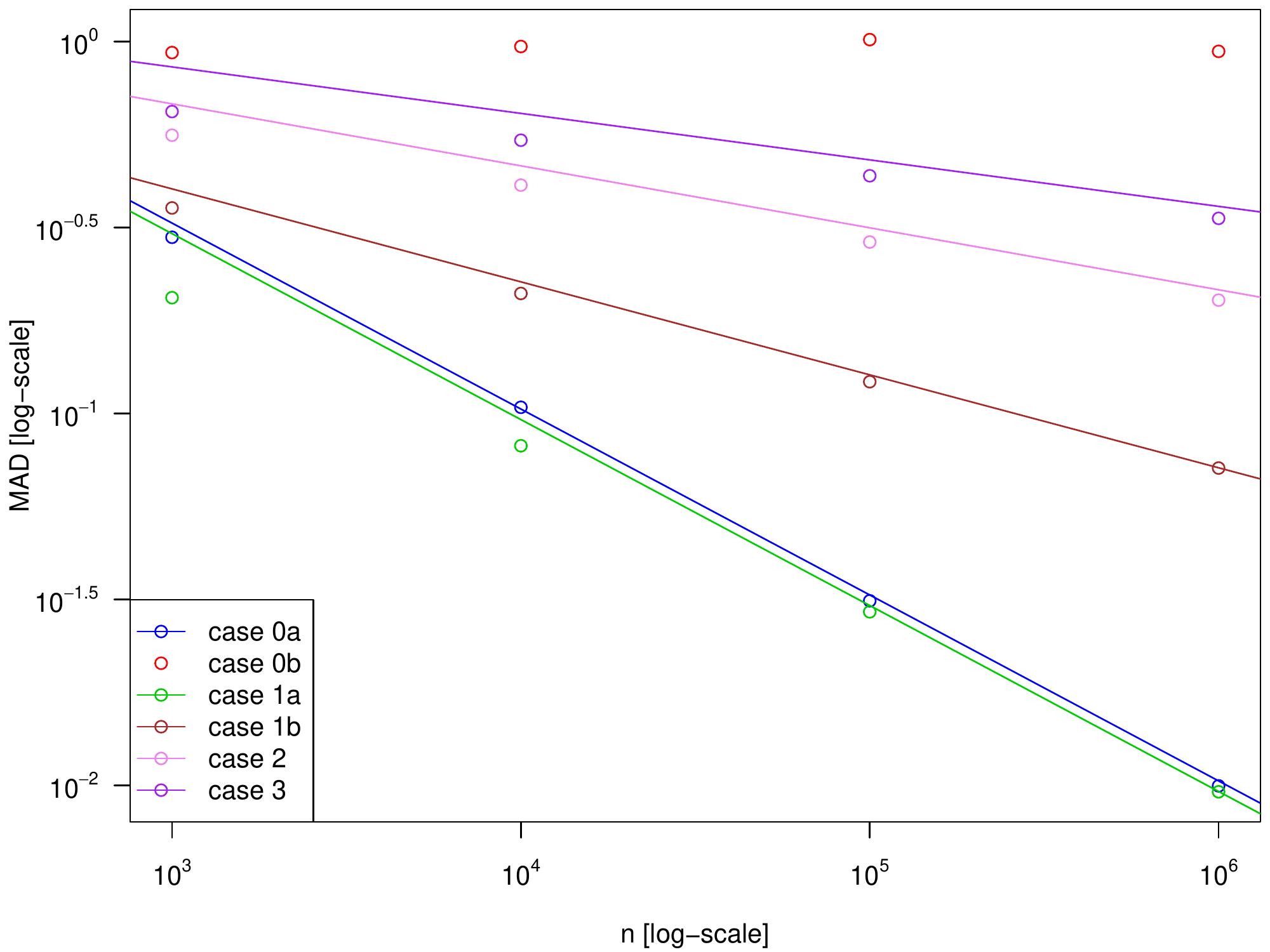}
\caption{\label{figmad}Median absolute deviations of the intrinsic mean of $n$ i.i.d. draws from the simulated distributions with parameters in Table~\ref{tabsim}, based on $1,000$ repetitions; lines give the values predicted using the asymptotic distribution.}
\end{figure}

\begin{figure}[!ptb]
\centering
\includegraphics[width=0.8\textwidth]{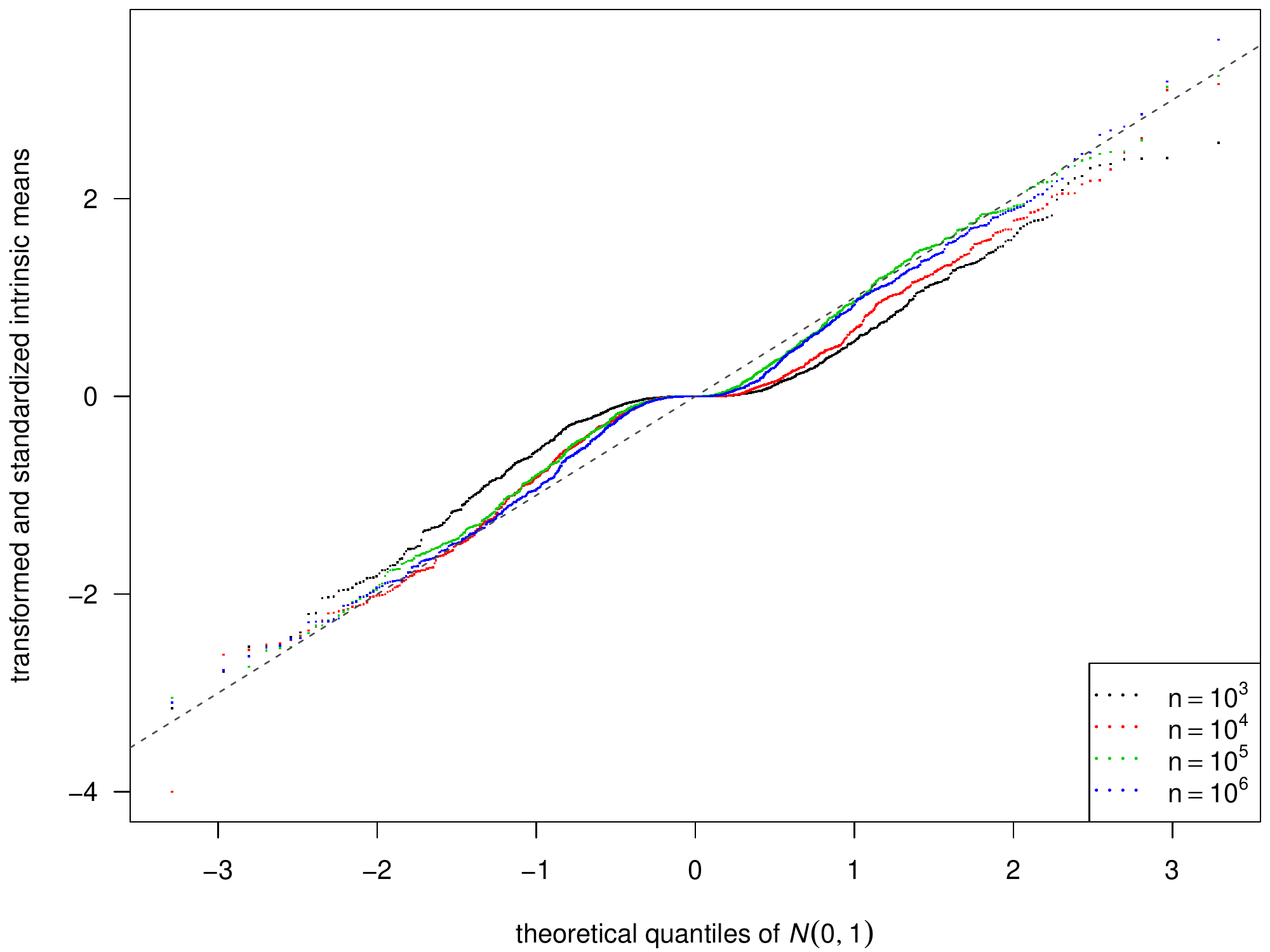}
\caption{\label{figqq}Q-q-plots of $1,000$ intrinsic means, transformed and standardized according to the asymptotic distribution, each based on $n$ i.i.d. draws from case 1b in Table~\ref{tabsim}; the dashed line depicts the identity line.}
\end{figure}

We also determined the median absolute deviation (MAD) of the intrinsic sample means for the different cases in Table~\ref{tabsim}, and compared them to the MAD predicted from the asymptotic distribution given in Theorem~\ref{CLT:thm} (except for case~0b where it does not apply), see Figure~\ref{figmad}.
For this, one easily computes
\begin{align*}
\sigma^2 &= 2 \int_0^{2\delta\pi} (\pi - x)^2 f(x) dx
= 2 \int_0^{\delta\pi} (\pi - x)^2 f(x) dx + 2 \int_{\delta\pi}^{2\delta\pi} (\pi - x)^2 f(x) dx
\\ &= -\frac{2\alpha}{6\pi} [(\pi - x)^3]_0^{\delta\pi}
- \frac{\alpha}{2\pi(\delta\pi)^k} [\tfrac{\pi^2}{k+1} x^{k+1} - \tfrac{2\pi}{k+2}x^{k+2} + \tfrac{1}{k+3}x^{k+3}]_0^{\delta\pi}
\\ &\quad + \frac{\alpha}{2\pi(\delta\pi)^k} [\tfrac{\pi^2(1-2\delta)^2}{k+1}x^{k+1} + \tfrac{2\pi(1-2\delta)}{k+2}x^{k+2} + \tfrac{1}{k+3}x^{k+3}]_0^{\delta\pi}
\\ &= \frac{\alpha\pi^2}{3} \big(1 - (1-\delta)^3 \big) - \frac{\alpha\delta\pi^2}{2(k+1)} \big( 1 - (1 - 2\delta)^2 \big) + \frac{\alpha(\delta\pi)^2}{k + 2} (2 - 2 \delta),
\end{align*}
as well as $f^{(k)}(0+) = - \alpha \frac{k!}{4\pi(\delta\pi)^k}$.
Note that we chose the MAD as it commutes with the power transforms in Theorem~\ref{CLT:thm}(ii), as opposed to the standard deviation. Thus, we found the rates predicted in Remark~\ref{rate-mu-n:cor}, namely $n^{-\frac{1}{2}}$ if $\alpha < 1$ and $n^{-\frac{1}{2(k+1)}}$ if $\alpha = 1$ and $k > 0$, to match the observed MAD in Figure~\ref{figmad} well.

Furthermore, for case~1a, Figure~\ref{figqq} shows normal q-q-plots for the intrinsic sample means, transformed and standardized according to their asymptotic distribution, i.e. for $\sqrt{n} \,\sign(\mu_n)\, \vert\mu_n\vert^{k+1} \frac{-2\pi f^{(k)}(0+)}{(k+1)!\sigma)}$. Note that there appears to be a peak at $0$, visible from the curve getting almost horizontal there, which decreases with increasing $n$.

\section{Discussion}

	Let us conclude with a discussion of our rather comprehensive results on locus, uniqueness, asymptotics and numerics for intrinsic circular means. In the past, there has been a fundamental mismatch between distributional and asymptotic theory on non-Euclidean manifolds. While a great variety of distributions for circular data had been developed which very well reflect the non-Euclidean topology, e.g. nowhere vanishing densities, the central limit theorem had only been available for distributions essentially restricted to a subset of Euclidean topology. On the circle we have eliminated this mismatch. In particular our results state that 
	\begin{quote}{\it the more the non-Euclidean topology is reflected by a probabilty distribution, i.e. the closer the distribution near the antipode is to the uniform distribution, the larger the deviation from Euclidean asymptotics.} \end{quote}
	We expect that similar results are valid for general manifolds where the antipode needs to be replaced by the \emph{cut locus} $C(\mu)$. Unless $C(\mu)$ carries positive mass, $V(\mu)$ is still differentiable at $\mu$, see \cite{Pn06}. From what we observed for the circle, we conjecture that $C(\mu)$ cannot carry mass if $\mu$ is a local minimizer. However, generalizing all results obtained here to arbitrary Riemannian manifolds is the subject of future research, but note that our results at once carry over to tori, them being cross products of circles.

\bibliographystyle{../../../BIB/elsart-harv}
\bibliography{../../../BIB/shape,../../../BIB/biomech}

\end{document}